\documentclass[a4paper,UKenglish]{lipics-v2019}

\usepackage{microtype}

\usepackage{amsmath,amssymb,graphicx}
\usepackage{mathtools}
\usepackage{enumerate} 
\usepackage{url}
\usepackage{charter}
\usepackage{graphicx}
\usepackage{verbatim}
\usepackage{tikz}
\usepackage{ wasysym }
\usepackage{tikz}

\newcommand{\FO}{\mathsf{FO}}
\newcommand{\ML}{\mathsf{ML}}
\newcommand{\TML}{\ensuremath{\mathsf{TML}}}
\newcommand{\PTML}{\mathsf{PTML}}
\newcommand{\FOML}{\mathsf{FOML}}

\newcommand{\PSPACE}{\ensuremath{\textsf{PSPACE}}}
\newcommand{\NEXP}{\ensuremath{\textsf{NEXPTIME}}}
\newcommand{\EXPSPACE}{\textsf{EXPSPACE}}

\newcommand{\Var}{\textsf{Var}}
\newcommand{\SF}{\textsf{SF}}
\newcommand{\FV}{\textsf{Fv}}
\newcommand{\md}{\textsf{md}}
\newcommand{\TrIC}{\textsf{Tr}_1}
\newcommand{\TrTP}{\textsf{Tr}_2}

\renewcommand{\implies}{\supset}

\newcommand{\Ps}{\mathcal{P}}

\renewcommand{\phi}{\varphi}

\newcommand{\DNF}{\mathsf{DNF}}

\newcommand{\literal}{\textsf{literal}}
\newcommand{\literals}{\textsf{literals}}

\newcommand{\module}{\textsf{module}}
\newcommand{\modules}{\textsf{modules}}
\newcommand{\Qsafe}{\textsf{quantifier-safe}}
\newcommand{\NNF}{\textsf{NNF}}
\newcommand{\FSNF}{\textsf{FSNF}}
\newcommand{\components}{\mathsf{C}}
\newcommand{\junk}{\#}

\newcommand{\type}{\textsf{type}}
\newcommand{\Wit}{\textsf{Wit}}

\newcommand{\AllDelta}{\boldsymbol{\delta}}
\newcommand{\AllPsi}{\boldsymbol{\psi}}
\newcommand{\AllExists}{\textsf{E}}

\newcommand{\Nat}{\mathcal{N}}

\newcommand{\C}{C}

\title{Two variable fragment of Term Modal Logic}

\author{Anantha Padmanabha}{Institute of Mathematical Sciences, HBNI, Chennai, India}{ananthap@imsc.res.in}{https://orcid.org/0000-0002-4265-5772}{}

\author{R Ramanujam}{Institute of Mathematical Sciences, HBNI, Chennai, India}{jam@imsc.res.in}{}{}

\authorrunning{Padmanabha, Ramanujam}

\Copyright{Anantha Padmanabha and R Ramanujam}

 

\category{}

\supplement{}

\funding{}

\acknowledgements{We thank Kamal Lodaya, Sreejith and Yanjing Wang for the insightful discussions.}

\EventEditors{Peter Rossmanith, Pinar Heggernes, and Joost-Pieter Katoen}
\EventNoEds{3}
\EventLongTitle{44th International Symposium on Mathematical Foundations of Computer Science (MFCS 2019)}
\EventShortTitle{MFCS 2019}
\EventAcronym{MFCS}
\EventYear{2019}
\EventDate{August 26--30, 2019}
\EventLocation{Aachen, Germany}
\EventLogo{}
\SeriesVolume{138}
\ArticleNo{36}

\nolinenumbers 
\hideLIPIcs  

\begin{document}

\maketitle

\begin{abstract}
Term modal logics ($\TML$) are modal logics with unboundedly many modalities, with
quantification over modal indices, so that we can have formulas of the form
$\exists y \forall x~ (\Box_x P(x,y) \implies \Diamond_y P(y,x))$. Like First order
modal logic, $\TML$ is also `notoriously' undecidable, in the sense that even
very simple fragments are undecidable. In this paper, we show the decidability
of one interesting  fragment, that of two variable $\TML$. This is in
contrast to two-variable First order modal logic, which is undecidable.

 \end{abstract}

\section{Introduction}
Propositional multi-modal logics ($\ML$) are extensively used in many areas 
of computer science and artifical intelligence (\cite{bluebook, hughes96}).
$\ML$ is built upon propositional logic by adding modal operators $\Box_i$ 
and $\Diamond_i$ for every index $i$ in a fixed finite set $Ag$ which is often
interpreted as a set of agents (or reasoners). Typically, the satisfiability 
problem is decidable for most instances of $\ML$. 

A natural question arises when we wish the set of modalities to be unbounded.
This is motivated by a range of applications such as client-server systems, 
dynamic networks of processes, games with unboundedly many players, etc. In
such systems, the number of agents is not fixed a priori. For some cases, the agent set can vary not
only across models, but also from state to state (ex. when new clients enter
the system or old clients exit the system).

\medskip

 Term Modal logic ($\TML$)  introduced by Fitting, 
Voronkov and Thalmann \cite{TML2001} addresses this requirement. $\TML$ is built upon first order logic, but 
the variables now range over modalities: so we can index the modality by terms 
$(\Box_x \alpha)$ and these terms can be quantified over. State assertions 
describe properties of these `agents'.  Thus we can write formulas of the form: 
$\forall x   (\Box_x P(x) \implies \exists y~ \Box_y \Diamond_x R(x,y))$. In \cite{PR17}
we have advocated $\PTML$, the propositional fragment of $\TML$, as a suitable logical
language for reasoning about systems with unboundedly many agents.  $\TML$ has 
been studied in dynamic epistemic contexts in \cite{kooi2007} and in modelling situations 
where the identity of agents is not common knowledge among the agents \cite{WangTML}. 

\medskip

The following examples illustrate the flavour of properties  that can be expressed in $\TML$.

\begin{itemize}
\item For every agent $x$ there is some agent $y$ such that $P(x,y)$ holds at all $x$-successors or there is some $y$-successor where $\neg P(x,y)$ holds.\\
 $\forall x \exists y~ \big( \Box_x P(x,y) \lor \Diamond_y (\neg P(x,y))\big)$
\item Every agent of type $A$ has a successor where some agent of type $B$ exists.\\
$\forall x \big(A(x) \implies \Diamond_x \exists y~ B(y)\big)$.

\item There is some agent $x$ such that for all  agents $y$ if there are no $y$ successors then in all successors of $x$, there is a $y$ successor. \\ 
$\exists x \forall y~ \big( \Box_y \bot \implies \Box_x \Diamond_y \top)$.
\end{itemize}

\bigskip

Since  $\TML$ contains first order logic, its satisfiability is clearly undecidable. 
We are then led to ask: can we build term modal logics over decidable fragments of 
first order logic? Natural candidates are the monadic fragment, the two-variable
fragment and the guarded fragment \cite{mortimer2var, guardedFO}.

$\TML$ itself can be seen as a fragment of first order modal logic ($\FOML$) 
\cite{FOMLbook} which is built upon first order logic by adding modal operators. 
There is a natural translation of $\TML$ into $\FOML$  by inductively
translating $\Box_x \alpha$ into $\Box (P(x) \implies \alpha)$ and $\Diamond_x \alpha$ 
into $\Diamond (P(x) \land \alpha)$ to get an equi-satisfiable formula, where $P$ is 
a new unary predicate. Sadly, this does not help much, since $\FOML$ is notorious for 
undecidability. The modal extension of many simple decidable fragments of first order 
logic become undecidable. For instance, the monadic fragment\cite{Kripke} or the two variable fragment \cite{FOML2var}
of $\FOML$ are undecidable. In fact $\FOML$ with two variables and a single unary predicate is already undecidable \cite{rybakov2017}. Analogously, in \cite{PR17} we
show that the satisfiability problem for $\TML$ is undecidable even when the atoms 
are restricted to propositions. In the presence of equality (even without propositions), this result can be 
further strengthened to show `Trakhtenbrot' like theorem of mutual recursive inseparability.

On the other hand, as we show in \cite{PR17}, the {\sf monodic} fragment of $\PTML$
(the propositional fragment) is decidable (a formula $\phi$ is monodic if each of its
modal subformulas of the form $\Box_x \psi$ or $\Diamond_x \psi$ has a restriction that the free variables of $\psi$ is contained in $\{x\}$). Further,
via the $\FOML$ translation above, we can show that the monodic restriction of $\TML$ based on
the guarded fragment of first order logic and monadic first order logic are decidable
\cite{FOMLdecidable}. 

In a different direction,  Wang (\cite{Wang17}) considered a fragment of 
$\FOML$ in which modalities and quantifiers are  bound to 
each other. In particular he considered the fragment with $\exists \Box$ and showed
it to be decidable in $\PSPACE$. In \cite{PRW18} it is proved that this technique of
{\sf bundling} quantifiers and modalities gives us interesting decidable fragments
of $\FOML$, and as a corollary, the bundled fragment of $\TML$ is decidable where 
quantifiers and modalities always occur in bundled form: $\forall x \Box_x \alpha, 
\exists x \Box_x \alpha$ and their duals. However, more general bundled fragments
of $\TML$ (such as those based on the guarded fragment of first order logic) have
been shown to be decidable by Orlandelli and Corsi (\cite{orlandelli}), and by
Shtakser (\cite{shtakser2018}). From all these results, it is clear that the one variable fragment of $\TML$
is decidable, and that the three variable fragment of $\PTML$ is undecidable.

\bigskip

 In this
paper, we show that the two variable fragment of $\TML$ ($\TML^2$) is decidable. This is in
contrast with $\FOML$, for which the two variable fragment is undecidable \cite{FOML2var}.  Quoting Wolter and Zakharyaschev from \cite{FOMLdecidable}, where they discuss the root of undecidability of $\FOML$ fragments:

\begin{quote} 
All undecidability proofs of modal predicate logics exploit formulas of the form $\Box~ \psi(x,y)$ in which the necessity operator applies to subformulas of more than one free variable; in fact, such formulas play an essential role in the reduction of undecidable problems to those fragments$\ldots$
\end{quote}

 Note that this is not  expressible in $\TML^2$ where there is no `free' modality;
every modality is bound  an index ($x$ or $y$). With a third variable $z$, we could indeed 
encode $\Box P(x,y)$ as $\forall z \Box_z P(x,y)$, but we do not have it. The
decidability of the two variable fragment of $\TML$, {\em without constants or 
equality}, hinges crucially on this lack of expressiveness. Thus, $\TML^2$ provides a decidable fragment of $\FOML^2$. From $\FO^2$ view point, Gradel and Otto\cite{FO2extension} show that most of the natural extensions of $\FO^2$ (like transitive closure, lfp) are undecidable except for the counting quantifiers.  In this sense, $2$-variable $\TML$ can be seen as another rare extension of $\FO^2$ that still remains decidable. Note that in this paper we consider the two variable
fragment of $\TML$ without the bundling or guarded or monodic restriction. Also, there is no natural translation of two variable $\TML$ to any known decidable fragment of $\FO$ such as the two variable fragment of $\FO$ with $2$ equivalence relations etc (cf \cite{Shtakser2018b}).

Thus, the contribution of this paper is technical, mainly in the identification
of a decidable fragment of $\TML$.  As is standard with two variable 
logics, we first introduce a normal form which is a combination of Fine's normal 
form for modal logics (\cite{ModalNormalForm}) and the Scott normal form (\cite{FO2})
for $\FO^2$. We then prove a {\sf bounded agent property} using an argument that
can be construed as modal depth induction over the `classical' bounded model construction
for $\FO^2$. 

\section{$\TML$ syntax and semantics}
We consider  relational vocabulary with {\sf no constants or function symbols, and
without equality}. 

\begin{definition}[$\TML$ syntax]
\label{def: TML syntax}
Given  a countable set of variables $\Var$ and a countable set of predicate symbols $\Ps$,  
the syntax of $\TML$ is defined as follows:
$$ \phi::= P(\overline{x}) \mid  \neg\phi \mid (\phi\land\phi) \mid (\phi \lor \phi) \mid \exists x~\phi \mid \forall x~\phi   \mid \Box_x \phi \mid \Diamond_x \phi$$
where $x \in \Var$, $\overline{x}$ is  a vector of length $n$ over $Var$ and $P\in \Ps$ of arity $n$. 
\end{definition}

  The \textit{free} and \textit{bound} occurrences of variables are defined 
as in $\FO$ with $\FV(\Box_x \phi) = \FV(\phi) \cup \{x\}$. We write $\phi(\overline{x})$ 
if all the free variables in $\phi$ are included in $\overline{x}$. 
Given a $\TML$ formula $\phi$ and $x, y\in \Var$, if $y \not\in \FV(\phi)$ then we write $\phi[y\slash x]$ for the 
formula obtained by replacing every occurrence of $x$ by $y$ in $\phi$. A formula $\phi$ is called 
a {\em sentence} if $\FV(\phi) = \emptyset$.  The notion of modal depth of 
a formula $\phi$ (denoted by $\md(\phi)$) is also standard, which is simply the maximum number of nested modalities occurring in $\phi$. The length of a formula $\phi$ is denoted by 
$|\phi|$ and is simply the number of symbols occurring in $\phi$.


In the semantics, the number of accessibility relations is not fixed, but 
specified along with the structure. Thus the Kripke frame for $\TML$ is given 
by $(W,D,R)$ where $W$ is a set of worlds, $D$ is the potential set of agents and 
$R\subseteq (W\times D \times W)$. The agent dynamics is captured by a function
($\delta: W \rightarrow 2^D$ below) that specifies, at any world $w$, the set of 
agents {\em live} (or meaningful) at $w$. The condition that whenever $(u,d,v)\in R$, we
have that $d \in \delta(u)$ ensures only an agent alive at $u$ can consider $v$ accessible.

A {\em monotonicity} condition  is imposed on the accessibility relation as well:
whenever $(u,d,v)\in R$, we have that $\delta(u)\subseteq \delta(v)$. This
is required to handle interpretations of free variables (cf \cite{corsi2002unified, TML2001, FOMLbook}). 
Hence the models are called `increasing agent' models.
 
\begin{definition}[$\TML$ structure]
\label{def: TML structure}
An increasing agent model for $\TML$ is defined as the tuple $M = (W, D, \delta, R, \rho)$ where
$W$ is a non-empty countable set of worlds, $D$ is a non-empty countable set of agents, 
$R\subseteq (W\times D\times W)$ and $\delta:W\to 2^D$. The map $\delta$ assigns to each 
$w\in W$ a \textit{non-empty} local domain such that whenever 
$(w,d,v) \in R$ we have $d\in \delta(w)\subseteq \delta(v)$ and 
$\rho: (W\times \Ps) \to \bigcup_{n\in \omega}2^{D^n}$ is the valuation function 
where for all $P \in \Ps$ of arity $n$ we have $\rho(w,P) \subseteq [\delta(w)]^n$.
\end{definition}

For a given model $M$, we use $W^M,D^M, \delta^M,R^M,\rho^M$ to refer to the 
corresponding components. We drop the superscript when $M$ is clear from the 
context. We often write $D_w$ for $\delta(w)$.  A \textit{constant agent} model 
is one where $D_w=D$ for all $w\in W$.  To interpret free variables, we  need a 
variable assignment $\sigma: \Var\to D$.  Call $\sigma$ {\em relevant} at 
$w \in W$ if $\sigma(x)\in \delta(w)$ for all $x\in \Var$. The increasing 
agent condition ensures that if $\sigma$ is relevant at $w$ and
 $(w,d,v)\in R$ then  $\sigma$ is relevant at $v$ as well. In a constant agent
model, every assignment $\sigma$ is relevant at all the worlds.

\begin{definition}[$\TML$ semantics]
\label{def: TML semantics}
Given a $\TML$ structure $M = (W, D, \delta, R, \rho)$ and a $\TML$ formula $\phi$, for all $w \in W$ and $\sigma$ relevant at $w$,  
define $M,w,\sigma \vDash \phi$ inductively as follows:
$$\begin{array}{|lcl|}
\hline
M, w, \sigma\vDash P(x_1,\ldots,x_n) &\Leftrightarrow & (\sigma(x_1), \ldots, \sigma(x_n))\in \rho(w,P)  \\ 
M, w, \sigma\vDash \neg\phi &\Leftrightarrow&   M, w, \sigma\nvDash \phi \\ 
M, w, \sigma\vDash (\phi\land \psi) &\Leftrightarrow&  M, w, \sigma\vDash \phi \text{ and } M, w, \sigma\vDash \psi \\ 
M, w, \sigma\vDash \exists x~ \phi &\Leftrightarrow& \text{there is some $d\in \delta(w)$ such that M, w, $\sigma_{[x\mapsto d]}\vDash\phi$ }\\
M, w, \sigma\vDash \Box_x~ \phi &\Leftrightarrow& M, v, \sigma\vDash\phi \text{ for all $v$ s.t.\ $(w,\sigma(x),v)\in R$}\\

\hline
\end{array}$$

\noindent where $\sigma_{[x\mapsto d]}$ denotes another assignment that is the same as $\sigma$ 
except for mapping $x$ to $d$. 
\end{definition}

The semantics for $\phi \lor \psi, \forall x ~\phi$ and $\Diamond_x~\phi$ are defined analogously. Note that $M,w,\sigma \vDash \phi$ is inductively defined only when $\sigma$ is 
relevant at $w$. We often abuse notation and say `for all $w$ and for all 
interpretations $\sigma$', when we mean `for all $w$ and for all interpretations 
$\sigma$ relevant at $w$' (and we will ensure that relevant $\sigma$ are used in
proofs).  In general, when considering the truth of $\phi$ in a model, it
suffices to consider $\sigma: \FV(\phi) \mapsto D$, assignment restricted to the variables
occurring free in $\phi$.  When $\FV(\phi) \subseteq \{x_1, \ldots, x_n\}$ and $\overline{d}
\in [D_w]^n$ is a vector of length $n$ over $D_w$, we  write $M, w\vDash \phi [\overline{d}]$ to denote $M,w,\sigma\vDash 
\phi(\overline{x})$ where  for all $i\le n,\ \sigma(x_i) = d_i$. 
When $\phi$ is a sentence, we  simply write $M,w \models \phi$. A formula $\phi$ is \textit{valid}, if $\phi$ is true in all models $M$ at all $w$ for all interpretations $\sigma$ (relevant
at $w$). A formula $\phi$ is \textit{satisfiable} if $\neg \phi$ is not valid.

\bigbreak

Now we take up the satisfiability problem which is the central theme of this paper. First we  observe that the satisfiability problem is equally hard for constant and increasing agent models for $\TML$.

First we prove  that the satisfiability problem over constant agent structures and increasing agent structures is equally hard for most fragments.
To see why this is true, if a formula $\phi \in \TML$ is satisfiable in some increasing agent model, 
then we can turn the model into constant agent model as follows. We introduce a new unary 
predicate $E$ and ensure that $E(d)$ is true at $w$ if $d$ is a member of $\delta(w)$ in
the given increasing agent model. But now, all quantifications have to be relativized with 
respect to the new predicate $E$. This translation is similar 
in approach to the one for $\FOML$\cite{FOMLdecidable}. The  syntactic translation is defined as follows:

\begin{definition}
\label{def: Inc to Const. translation}
Let $\phi$ be any $\TML$ formula and let $E$ be a new unary predicate  not occurring in $\phi$. The translation  is defined inductively as follows:
\begin{itemize}
\item $\TrIC(P(x_1, \ldots, x_n)) = P(x_1, \ldots, x_n)$
\item $\TrIC(\neg \phi) = \neg \TrIC(\phi)$ and $\TrIC(\phi \land \psi) = \TrIC(\phi) \land \TrIC(\psi)$
\item $\TrIC(\Box_x \phi) = \Box_x ( \TrIC(\phi))$
\item $\TrIC(\exists x~\phi)  = \exists x~ (E(x) \land \TrIC(\phi))$
\end{itemize}
\end{definition} 

With this translation, we also need to ensure that the predicate $E$ respects monotonicity. 
Hence we have 
$\gamma_\phi = \bigwedge\limits_{i+j\le\md(\phi)} (\forall y \Box_y)^i \big(\forall x~ E(x) 
\implies (\forall y\Box_y)^j E(x)\big)$. Now, we can prove that $\phi$ is 
satisfiable in an increasing model iff $\TrIC(\phi) \land \gamma_\phi$ 
is satisfiable in a constant agent model. Moreover, both the formulas are satisfiable over the same agent set $D$.

\begin{lemma}
\label{lem: inc and const domain sat are same}
Let $\phi$ be any $\TML$ formula.  $\phi$ is satisfiable in an increasing agent model with agent set $D$ iff $ \gamma_\phi \land \TrIC(\phi)$ is satisfiable in a constant agent model with agent set $D$.
\end{lemma}

\begin{proof}

$(\Rightarrow)$ Suppose $M^I = (W,D,\delta^I,R,\rho^I)$ is an increasing agent model with $r\in W$ such that $M^I,r,\sigma \models \phi$.  Define the constant domain model $M^C = (W,D,\delta^C,R,\rho^C)$ where $\delta^C(w) = D$ for all $w\in W$ and  $\rho^C$ is the same as $\rho^I$ for all predicates except $E$ and for all $w\in W$ and $d\in D$ we have $d\in \rho^C(w,E)$ iff $d \in \delta(w)$.

 Since $\delta^I$ is monotone, $M^C,r,\sigma \models \gamma_\phi$. Note that $M^I,w \models  P(\overline{d})$ iff $M^C,w \models P(\overline{d})$ and we have $d \in \delta^I(w)$ iff $M^C,w \models E(d)$. Thus, we can set up a routine induction and prove that for all subformulas $\psi$ of $\phi$ and for all $w\in W$  and for all  interpretation $\sigma'$ relevant at $w$, we have  $M^I,w,\sigma' \models \psi$  iff $M^C,w,\sigma' \models \TrIC(\psi)$. Hence, $M^C,r,\sigma \models \TrIC(\phi)$.

\bigskip

$(\Leftarrow)$ Suppose $M^C = (W,D,\delta^C,R,\rho^C)$ is  a tree model of depth at most $\md(\phi)$ with $r\in W$ such that $M^C,r,\sigma \models \gamma_\phi \land \TrIC(\phi)$. Define the increasing agent model $M^I = (W,D,\delta^I,R,\rho)$ where $c\in \delta^I(w)$ iff $M,w \models E(c)$.

Note that $\delta^I$ defined above is monotone since $M^C,r \models \gamma_\phi$. Again, we can set up a routine induction and prove that for all subformulas $\psi$ of $\phi$ and for all $w\in W$  and for all interpretation $\sigma'$ relevant at $w$ we have  $M^C,w,\sigma' \models \TrIC(\psi)$  iff $M^I,w,\sigma' \models \psi$.

\end{proof}

\bigbreak

The propositional term modal logic $(\PTML)$ is a fragment of $\TML$ where the atoms are restricted to propositions. Note that the variables  still appear as index of modalities. For $\PTML$, the valuation function can be simply written as $\rho : W \mapsto 2^\Ps$ where $\Ps$ is the set of propositions. 
Now we prove that the satisfiability problem for $\PTML$ is as hard as that for $\TML$.
The reduction is based on the translation of an arbitrary atomic predicate $P(x_1,\ldots, x_n)$ to $\Diamond_{x_1}\ldots \Diamond_{x_n} p$ where $p$ is a new proposition which represents the predicate $P$.  However, this cannot be used always\footnote{for instance, this translation will not work for the formula $\exists x~ P(x) \land \forall y~ \Box_y \bot$}. Thus, we use a new proposition $q$, to distinguish the `real worlds' from the ones that are added because of the translation. But now, the modal formulas have to be relativized with respect to the proposition $q$. The formal translation is given as follows:

\begin{definition}
\label{def: TML to PTML translation for satisfiability}
Let $\phi$ be any $\TML$ formula where $P_1,\ldots, P_m$ are the predicates that occur in $\phi$. Let $\{p_1,\ldots ,p_m\} \cup \{ q\}$ be a new set of propositions not occurring in $\phi$. The translation  with respect to $q$ is defined inductively as follows:
\begin{itemize}
\item $\TrTP(P_i(x_1,\ldots, x_n); q) = \Diamond_{x_1} (\neg q \land \Diamond_{x_2}(\ldots \neg q \land \Diamond_{x_n} (\neg q \land p_i) \ldots))$
\item $\TrTP (\neg \phi; q) = \neg \TrTP (\phi; q)$ and $\TrTP (\phi \land \psi ; q) = \TrTP(\phi; q) \land \TrTP(\psi; q)$
\item $\TrTP(\Box_x \phi; q) = \Box_x (q \implies \TrTP(\phi; q))$
\item $\TrTP (\exists x~\phi; q)  = \exists x~ \TrTP(\phi; q)$
\end{itemize}
\end{definition}

\begin{lemma}
\label{lem: TML and PTML satisfiability are same}
For any $\TML$ formula $\phi$, we have $\phi$ is satisfiable in an increasing (constant) agent model with agent set $D$ iff $q \land \TrTP(\phi; q)$ is satisfiable in an increasing (constant)  agent model with agent set $D$.
\end{lemma}

\begin{proof} Let  $P_\phi$ be the set of all predicates occurring in $\phi$ and $k$ be the maximum arity among the predicates in $P_\phi$.  For any model $M$ and $u\in W$ let $\overline{c} \in D_u^*$ denote a (possibly empty) string of finite length over $D_u$.
 
 \medskip

$(\Rightarrow)$ Suppose the $\TML$ formula $\phi$ is satisfiable. Let $M^T = (W^T,D,\delta^T,R^T,\rho^T)$ be a $\TML$ model and $w\in W$ such that $M^T,w,\sigma \models \phi$. Define the $\PTML$ model $M^P = (W^P,D,\delta^P,R^P,val^P)$ where:
\begin{itemize}
\item[-] $W^P = \{ u_{\overline{c}}\mid u\in W^T$ and $\overline{c} \in D^*_u$ of length at most $k\}$. 
\item[-] For all $u_{\overline{c}} \in W^P$ we have $\delta^P(u_{\overline{c}}) = \delta^T(u)$.
\item[-] $R^P = \{ (u_\epsilon, c, v_\epsilon) \mid (u,c,v)\in R^T\} \cup \{ (u_{\overline{c}},d,u_{\overline{c}d}) \mid u_{\overline{c}}, u_{\overline{c}d} \in W^P\}$
\item[-] $\rho^P(u_\epsilon) = \{ s \mid s$ is a proposition in $P_\phi$ and $M,u \models s\} \cup \{q\}$ and\\ $\rho^P(u_{c_1\ldots c_n}) = \{ p_i \mid M,u \models P_i(c_1,\ldots ,c_n)\}$.
\end{itemize}

 Note that $M^T,u,\sigma \models P_i(x_1,\ldots, x_n)$ iff $M^P,u_\epsilon,\sigma \models \Diamond_{x_1} (\neg q \land \Diamond_{x_2}(\ldots \neg q \land \Diamond_{x_n} (\neg q \land p_i) \ldots))$ and for all $u\in W^T$ we have $M,u_\epsilon \models q$. Thus a standard inductive argument shows that for all subformulas $\psi$ of $\phi$ and for all $u\in W^T$ and for all interpretation $\sigma'$ we have $M^T,u,\sigma' \models \psi$ iff $M^P,u_\epsilon,\sigma' \models r \land \TrTP(\psi)$.
 
 Also note that if $M^T$ is an increasing (constant) agent model over $D$ then $M^P$ is also an increasing (constant) agent model over $D$.

\bigskip

$(\Leftarrow)$ Suppose $M^P = (W^P,D,\delta^P,R^P,\rho^P)$ such that $M,w\models r\land \TrTP(\phi)$. Define $M^T = (W^T,D,\delta^T,R^T,\rho^T)$ where
\begin{itemize}
\item[-] $W^T = \{ u \in W^P \mid M^P,u \models q\}$. 
\item[-] For all $u\in W^T$ we have $\delta^T(u) = \delta^P(u)$.
\item[-] $R^T = R^P \cap (W^T \times W^T)$.
\item[-] $\rho^T(u,P) = \{ (c_1,\ldots, c_n) \mid M^P,u \models \Diamond_{c_1} (\neg q \land(\ldots \Diamond_{c_n} (\neg q \land p)\}$ and\\ $q\in \rho^T(u)$ iff $q\in \rho^P(u)$.
\end{itemize}

Note that for all $u\in W^P$ we have $M^P,u \models \Diamond_{c_1} (\neg q \land(\ldots \Diamond_{c_n} (\neg q \land p_i)$ iff $M^T,w \models P_i(c_1,\ldots, c_n)$. Also, since $M^P,w \models q$ we have $w\in W^T$. Again, an inductive argument shows that for all subformulas $\psi$ of $\phi$ and for all $u\in W^T$ and for all interpretation $\sigma'$ relevant at $w$, we have $M^P,u,\sigma' \models \TrTP(\psi)$ iff $M^T,u,\sigma' \models \psi$. Thus we have $M^T,w,\sigma \models \phi$.

To complete the proof, again note that if $M^P$ is an increasing (constant) agent model over $D$ then $M^T$ is also an increasing(constant) agent model over $D$.
\end{proof}

\section{Two variable fragment}
Note that all the examples discussed in the introduction section use only 2 variables. Thus, 
$\TML$ can express interesting properties even when restricted to two variables.
We now consider the satisfiability problem of $\TML^2$.
The translation in Def. \ref{def: TML to PTML translation for satisfiability}
preserves the number of variables. Therefore it suffices to consider the satisfiability 
problem for the two variable fragment of $\PTML$.

Let $\PTML^2$ denote the two variable fragment of $\PTML$.  We first consider a normal form
for the logic.  In \cite{ModalNormalForm}, Fine introduces a normal form for propositional
modal logics which is a disjunctive normal form ($\DNF$) with every clause of the form 
 $(\bigwedge\limits_{i}(s_i) \land \Box \alpha \land \bigwedge\limits_j \Diamond \beta_j)$ 
where $s_i$ are literals and $\alpha, \beta_j$ are again in the normal form. For $\FO^2$, we 
have Scott normal form \cite{FO2} where every $\FO^2$  sentence has an equi-satisfiable 
formula of the form 
$\forall x \forall y ~\phi \land \bigwedge\limits_i\forall x \exists y~ \psi_i$ 
where $\phi$ and $\psi_i$ are all quantifier free. 
For $\PTML^2$, we introduce a combination of these two normal forms, which we call the \emph{Fine Scott Normal form} given by a $\DNF$, where every clause is of the form:
 
 $$\bigwedge\limits_{i\le a} s_i \land
 \bigwedge\limits_{z\in \{x,y\}} (\Box_z \alpha \land \bigwedge\limits_{j\le m_z}\Diamond_z \beta_j )\ \land\
  \bigwedge\limits_{z\in \{x,y\}} (\forall z~ \gamma \land \bigwedge\limits_{k\le n_z} \exists z~ \delta_k)\ \land\
 \forall x \forall y~\phi \land \bigwedge\limits_{l\le b}\forall x \exists y~ \psi_l$$ 
where $a,m_x,m_y,n_x,n_y,b \ge 0$ and $s_i$ denotes literals. Further, $\alpha,\beta_j$ are recursively in the normal form and $\gamma, \delta_k, \phi, \psi_l$ do not have quantifiers at the outermost level and all modal subformulas occurring in these formulas are (recursively) in the normal form. The normal form is formally defined in the next subsection.

Note that the first two conjuncts mimic the modal normal form and the last two conjuncts mimic the $\FO^2$ normal form. The additional conjuncts handle the intermediate step where only one of the variable is quantified and the other is free.

We now formally define the normal form and prove that every $\PTML^2$ formula has a corresponding 
equi-satisfiable formula in the normal form. After this we prove the bounded agent 
property for formulas in  the normal form using an inductive $\FO^2$ type model construction. 
 
\subsection{Normal form}
We use $\{x,y\} \subseteq \Var$ as the two variables of $\PTML^2$. We use $z$ to refer to either $x$ or $y$ and  refer to variables $z_1,z_2$ to indicate the variables $x,y$ in either order.  We use $\Delta_z$ to denote any modal operator $\Delta \in \{ \Box,\Diamond\}$ and $z\in \{x,y\}$. A $\literal$ is either a proposition  or its negation. Also, we assume that the formulas are given in negation normal form($\NNF$) where the negations are pushed in to the $\literals$. 

\begin{definition} [$\FSNF$ normal form]
\label{def: FSNF normal form}
We define the following terms to introduce the {\em Fine Scott normal form} ($\FSNF$) for $\PTML^2$:
\label{def: normal form}  
\begin{itemize}
\item A formula $\phi$ is a $\module$ if $\phi$ is a $\literal$ or $\phi$ is of the form $\Delta_z \alpha$. 
\item For any formula $\phi$, the outer most components of $\phi$ given by $\components(\phi)$ is defined inductively where for any $\phi$ which is a $\module$, $\components(\phi) = \{ \phi\}$ and $\components(Qz\ \phi) = \{ Qz \ \phi\}$ where $z\in \{x,y\}$ and $Q\in \{\forall,\exists\}$. Finally $\components(\phi \odot \psi) = \components(\phi) \cup \components(\psi)$ where $\odot \in \{ \land, \lor\}$.

\item A formula $\phi$ is  $\Qsafe$  if  every $\psi \in \components(\phi)$ is a $\module$.

\item We define $Fine~Scott~normal~form$($\FSNF$) normal form ($\DNF$ and conjunctions) inductively as follows:
\begin{itemize}
\item  Any conjunction of $\literals$ is an $\FSNF$ conjunction.
\item $\phi$ is said to be in $\FSNF~\DNF$ if $\phi$ is a  disjunction where every clause is an $\FSNF$ conjunction.
\item Suppose $\phi$ is $\Qsafe$ and for every $\Delta_z \psi \in \components(\phi)$ if $\psi$ is in $\FSNF~\DNF$ normal form then we call $\phi$ a $\Qsafe~$normal formula.
\item Let $a,b,m_x,m_y,n_x,n_y \ge 0$. \\ Suppose $s_1,\ldots,s_a$ are $\literal$s, $\alpha^x,\alpha^y, \beta^x_1,\ldots ,\beta^x_{m_x}, \beta^y_1,\ldots, \beta^y_{m_y}$ are formulas in $\FSNF~\DNF$ and $\gamma^x,\gamma^y, \delta^x_1,\ldots, \delta^x_{n_x},\delta^y_1,\ldots ,\delta^y_{n_y}, \phi, \psi_1,\ldots, \psi_b$ are $\Qsafe~$ normal formulas then:

$$\bigwedge\limits_{i\le a} s_i \land
 \bigwedge\limits_{z\in \{x,y\}} (\Box_z \alpha^z \land \bigwedge\limits_{j\le m_z}\Diamond_z \beta^z_j )\ \land\
  \bigwedge\limits_{z_1\in \{x,y\}} (\forall z_2~ \gamma^{z_1}\ \land \bigwedge\limits_{k\le n_z} \exists z_2~\delta^{z_1}_k)\ \land\
 \forall x \forall y~ \phi \land \bigwedge\limits_{l\le b}\forall x \exists y~ \psi_l$$ 
is an $\FSNF$ conjunction.

\end{itemize}

\end{itemize}
\end{definition}

Quantifier-safe formulas are those in which no quantifiers occur outside the scope of modalities.  Note that the superscripts in $\alpha^x,\alpha^y$ etc  only  indicate which variable the formula is associated with, so that it simplifies the notation. For instance, $\alpha^x$ does not say anything about the free variables in $\alpha^x$. In fact there is no restriction on free variables in any of these formulas.

Further, note that by setting the appropriate indices to $0$, we can have $\FSNF$ conjunctions where one or more of the components corresponding to $s_i, \beta^x,\beta^y, \delta^x,\delta^y, \psi_l$ are absent. We also consider the conjunctions where one or more of the components corresponding to  $\Box_x \alpha^x, \Box_y \alpha^y, \phi$ are also absent. As we will see in the next lemma, for any sentence $\phi \in \PTML^2$, we can obtain an equi-satisfiable sentence, which  at the outer most level, is a $\DNF$ where every clause is of the form $\bigwedge\limits_{i\le a} s_i~ \land~
 \forall x \forall y~ \phi \land \bigwedge\limits_{l\le b}\forall x \exists y~ \psi_l$.

\begin{lemma}
\label{lem: FSNF equisat}
For every formula $\phi \in \PTML^2$ there is a corresponding formula $\psi$ in $\FSNF~\DNF$ such that $\phi$ and $\psi$ are equi-satisfiable.
\end{lemma}

\begin{proof}
We prove this by induction on the modal depth of $\theta$. Suppose $\theta$ has modal depth $0$, then all $\modules$ occurring in $\phi$ are $\literals$. Observe that if $\alpha$ is a propositional formula then for  $Q \in \{ \forall , \exists\}$ and $z\in \{x,y\}$ and for all model $M$  we have $M,w,\sigma \models Qz~\alpha$  iff $M,w,\sigma \models\alpha$. Hence we can simply ignore all the quantifiers and get an equivalent $\DNF$ over $\literals$, which is an $\FSNF~\DNF$.

For the induction step,  suppose $\md(\theta) = h$. First observe that  we can get an equivalent $\DNF$ formula for $\theta$ (say $\theta_1$)  over $\components(\theta)$ using propositional validities. Now if $\theta_1$ is an $\FSNF~\DNF$ then we are done. Otherwise, there are some clauses in $\theta_1$ that are not $\FSNF$ clause. Let $\theta_1 := \bigvee_i \zeta_i$  and  $I_\theta = \{ \zeta_i \mid \zeta_i$ is not a $\FSNF$ clause$\}$ be the clauses that are not $\FSNF$ conjunctions. To reduce $\theta_1$ in to $\FSNF~\DNF$, we  replace every $\zeta_i \in I_\theta$ with their corresponding equi-satisfiable $\FSNF~\DNF$ in $\theta_1$.

Pick a clause  $\zeta \in I_\theta$ and let $\zeta := \omega_1\land\ldots  \land \omega_n$ that is not an $\FSNF$ conjunction. If $\md(\zeta) < h$ then by induction hypothesis, there is an equi-satisfiable  $\FSNF~\DNF$ formula of $\zeta$. Thus $\zeta$ can be replaced by its corresponding equi-satisfiable $\FSNF~\DNF$ in $\theta_1$. Now suppose $\md(\zeta) = h$.  Call each $\omega_i$ as a conjunct.

 In the first step,  consider the conjuncts with exactly $1$ free variable. Let $I_z = \{ \omega_i \mid \FV(\omega) = \{z\}\}$ for $z\in \{x,y\}$ be the index of all conjuncts where $z$ is the only free variable.  Let $z_1,z_2$ be the variables $x,y$ in either order.  Pick any $\omega_i \in I_{z_1}$ which means $z_2$ is bounded in $\omega_i$. Hence, without loss of generality, $\omega_i$ is of the form $\forall z_2~\eta$. We will first ensure that $\eta$ is  $\Qsafe$. This is done by iteratively removing the non-$\modules$ from $\components(\eta)$ and replacing it with a equi-satisfiable $\Qsafe$ formula.  Set $\chi_0 := \forall z_2~ \eta$. 
 \begin{enumerate}
  \item[a.] if there  is some strict subformula of the form $Q z_2\ \lambda \in \components(\chi_0)$ where $\lambda$ is $\Qsafe$,  let $P$ be a new (intermediate) unary predicate. Define $\chi_1 :=  \chi_0[P(z_1) / Q z_2~\lambda]$ and $\tau_1 := P(z_1) \Leftrightarrow Q z_2~ \lambda$. Note that if $Q = \forall$ then $\tau_1$ can be equivalently written as  $\forall z_2~ ( \neg P(z_1) \lor \lambda) \land \exists z_2~ ( P(z_1) \lor \neg \lambda)$ and if $Q = \exists$ then $\tau_1$ will be $\exists z_2~ (\neg P(z_1) \lor \lambda) \land \forall z_2~ (P(z_1) \lor \neg \lambda)$. 
 \item[b.] if there is some strict subformula  of the form $Q z_1~ \lambda \in \components(\chi_0)$ where $\lambda$ is $\Qsafe$, let $P$ be a new unary predicate. Define $\chi_1 :=  \chi_0[P(z_2) / Q z_1~ \lambda]$ and $\tau_1 := \forall z_2~ (P(z_2) \Leftrightarrow Q z_1~\lambda)$. Again, that if $Q = \forall$ then $\tau_1$ is equivalent to  $\forall z_2 \forall z_1 ( \neg P(z_2) \lor \lambda) \land \forall z_2 \exists z_1 ( P(z_2) \lor \neg \lambda)$ and if $Q = \exists$ then $\tau_1$ is $ \forall z_2 \forall z_1 ~ (P(z_2) \lor \neg \lambda)\land \forall z_2 \exists z_1~ (\neg P(z_2) \lor \lambda)$.
\end{enumerate}

 Now remove the conjunct $\omega_i$ from $\zeta$ and replace it with $\chi_1 \land \tau_1$. Note that $\chi_1$ has at least one less quantifier than $\chi_0$ and $\tau_1$ introduces either  conjuncts with no free variables or a formula with one free variable of the form $Q z~\lambda$ where $\lambda$ is $\Qsafe$. To see that this step preserves equi-satisfiability, note that in both cases, $\chi_1 \land \tau_1$ implies $\chi_0$ and for the other direction, we can define the valuation $\rho$ for the new unary predicate $P$ appropriately in the given model in which $\psi$ is satisfiable.

 Repeat this step for $\chi_1, \chi_2 ,\ldots, \chi_m$ till $\chi_m$ is of the form $\forall z_2 \lambda$ where $\lambda$ is $\Qsafe$. Then we would have $\chi_m \land \tau_1\ldots\land \tau_m$ as new conjuncts replacing $\omega_i$ in $\zeta$.  Now this step increases the number of conjuncts in $\zeta$ which have no free variables, but all new conjuncts with one free variable is of the form $Q z~ \lambda$ where $\lambda$ is $\Qsafe$ (it needs to be further refined since it is not yet $\Qsafe~\FSNF$). 
 
 Repeat this step for all $\omega_i \in I_z$ for $z\in \{x,y\}$. Let the resulting clause be $\zeta_1$ which is equi-satisfiable to $\zeta$. Now for $z\in \{x,y\}$, if there are two conjuncts of the form $\forall z ~\lambda$ and  $\forall z~ \lambda'$ in $\zeta_1$, remove both of them from  and add $\forall z ~ (\lambda \land \lambda')$ to $\zeta_1$. Repeat this till there a single conjunct in $\zeta_1$ of the form $\forall z~ \gamma^z$ for each $z\in \{x,y\}$ where $\gamma^z$ is $\Qsafe$.  Note that there are some new unary predicates introduced and hence this intermediate formula $\zeta_1$ is not in $\PTML^2$ (but is in $\TML^2$). 
 
 \bigskip
Let $\zeta_1:= \omega'_1 \land \ldots\land \omega'_{n_1}$ which is the result of rewriting of the clause $\zeta$ after the above steps.  Now  consider conjuncts with no free variables and make them $\Qsafe$. Let $I = \{\omega'_i \mid \FV(\psi') = \{ x,y\}\}$. For any $\omega'_i \in I$, since neither variable is free, without loss of generality assume that $\omega'_i$ is of the form $\forall x~ \eta$.

Pick any $\omega'_i\in I$ and set $\chi_0 := \forall x ~\eta$ and  $z_1,z_2$ refer to $x,y$ in either order. If  $Q z_2~ \lambda \in \components(\eta)$, let $P$ be a new unary predicate. 
Define $\chi_1 := \chi_0[P(z_1) / Q z_2~ \lambda]$ and $\tau_1 := \forall z_1~ (P(z_1) \Leftrightarrow Qz_2~ \lambda$). Similar to previous step, $\tau_1$ can be equivalently written as two conjuncts of the form $\forall z_1 \forall z_2 ~\lambda \land \forall z_1 \exists z_2~ \lambda$ where $\lambda$ and $\lambda'$ are $\Qsafe$ formulas (but not $\Qsafe~\FSNF$, yet). 

 Now remove the conjunct $\omega'_i$ from $\zeta_1$ and replace it with $\chi_1 \land \tau_1$. Note that $\chi_1$ has at least one less quantifier than $\chi_0$ and $\tau_1$ introduces only  conjuncts of the form $Q_1 z_1 Q_2~z_2~\lambda$ where $\lambda$ is $\Qsafe$. Again for the equi-satisfiability argument, note that $\chi_1 \land \tau_1 \implies \chi_0$ is a validity and for the other direction, the new predicates can be interpreted appropriately in the same model of $\zeta_1$.

 Repeat this step for $\chi_1, \chi_2 ,\ldots, \chi_m$ till $\chi_m$ is of the form $\forall x \lambda$ where $\lambda$ is $\Qsafe$. Then we would have $\chi_m \land \tau_1\ldots\land \tau_m$ as new conjuncts replacing $\omega'_i$.   
 Now rename the variables appropriately in the newly introduced conjuncts so that we have formulas only of the form $\forall x \forall y~ \lambda$ or $\forall x\exists y~ \lambda'$ where $\lambda, \lambda'$ are  $\Qsafe$ formulas. 
 
 Repeat this step for all $\omega'_i \in I$. Let the resulting conjunct be $\zeta_2$ which is equi-satisfiable to $\zeta_1$.  Now if there are two conjuncts of the form $\forall x \forall y ~\lambda$ and  $\forall x \forall y ~\lambda'$ in $\zeta_2$, remove both of them and add a new conjunct $\forall x \forall y~ (\lambda \land \lambda')$ to $\zeta_2$. Repeat this till at most one conjunct  the form $\forall x\forall y~ \lambda$ in $\zeta_2$.  Note that we still have unary predicates in $\zeta_2$ and hence $\zeta_2$ is also a $\TML^2$ formula but not a $\PTML^2$ formula. Further, all subformulas inside the scope of quantifiers are now $\Qsafe$, but needs to be converted into $\Qsafe~\FSNF$.

\bigskip
Let $\zeta_2:= \omega''_1 \land \ldots\land \omega''_{n_2}$ be the resulting formula after the above steps. Now to eliminate the newly introduced unary predicates, apply the translation in definition \ref{def: TML to PTML translation for satisfiability} to $\zeta_2$ and  obtain an equi-satisfiable $\PTML$ formula $\zeta_3$. It is clear from the construction that the new predicates are introduced only at the outermost level (not inside the scope of any modality).  Thus, in the translation occurrence of the newly introduced predicate of the form $P(z)$ will be replaced by $\Diamond_z (\neg r \land p)$ and $\neg P(z)$ will be translated to $\neg \Diamond_z (\neg r \land p)$ which can be equivalently written as $\Box_z (r \lor \neg p)$.

\bigskip
Now consider conjuncts that are modal formulas. For $z\in \{x,y\}$, if there are two conjuncts of the form $\Box_z ~\lambda$ and  $\Box_z~ \lambda'$ in $\zeta_3$, remove both of them from  and add $\Box_z ~ (\lambda \land \lambda')$ to $\zeta_3$. Repeat this till there at most one conjunct in $\zeta_3$ of the form $\Box_z~ \alpha^z$ for each $z\in \{x,y\}$. 
Note that this step preserves equi-satisfiability because of the validity $\forall z~\big((\Box_z \alpha \land \Box_z \beta) \Leftrightarrow \Box_z (\alpha \land \beta)\big)$.

By rearranging the conjuncts, we obtain the formula $\zeta_3$ in the form:
$$\bigwedge\limits_{i\le a} s_i \land
 \bigwedge\limits_{z\in \{x,y\}} (\Box_z \alpha^z \land \bigwedge\limits_{j\le m_z}\Diamond_z \beta^z_j )\ \land\
  \bigwedge\limits_{z\in \{x,y\}} (\forall z~ \gamma^z \land \bigwedge\limits_{k\le n_z} \exists z\ \delta^z_k)\ \land\
 \forall x \forall y\ \phi \land \bigwedge\limits_{l\le b}\forall x \exists y\ \psi_l$$ 

where $\gamma, \delta^z_k, \phi$ and $\psi_l$ are $\Qsafe$. 

\bigskip

As a final step, we need to ensure that
$\alpha^x,\alpha^y, \beta^x_1,\ldots ,\beta^x_{m_x}, \beta^y_1,\ldots, \beta^y_{m_y}$ are formulas in $\FSNF~\DNF$ and $\gamma^x,\gamma^y, \delta^x_{1},\ldots, \delta^x_{n_x},\delta^y_1,\ldots ,\delta^y_{n_y}, \phi, \psi_1,\ldots, \psi_b$ are not just $\Qsafe$, but also $\Qsafe~\FSNF$ formulas.

 Note  $\alpha^z,\beta^z_j$ have modal depth less than $h$. Hence, inductively we have equi-satisfiable $\FSNF~\DNF$  which each of them can be correspondingly replaced  in $\zeta_3$.
 This preserves equi-satisfiability since we can inductive maintain that the translated formulas are satisfied in the same model of the given formula by just by tweaking the $\rho$ function.

To translate the formulas $\gamma^x,\gamma^y, \delta^x_{1},\ldots, \delta^x_{n_x},\delta^y_1,\ldots ,\delta^y_{n_y}, \phi, \psi_1,\ldots, \psi_b$ into $\Qsafe~\FSNF$, first note that these formulas are already $\Qsafe$. Now for every $\Delta_z \chi \in \components(\mu)$ for $\mu$ is one of the above formulas, we have $\md(\chi) \le h$. Again,  inductively we have equi-satisfiable $\FSNF$ formulas for each of them. Replacing each such subformula with their corresponding $\FSNF~\DNF$ formula gives us the required $\FSNF$ conjunction $\zeta_4$ which is equi-satisfiable to $\zeta$ that we started with. Thus $\zeta$ can be replaced by $\zeta_4$ in $\theta_1$.

\bigskip
Repeating this for every $\zeta \in I_\theta$ and replacing it in $\theta_1$ we obtain an  equi-satisfiable $\FSNF~\DNF$ for $\theta$.
\end{proof}

Since we repeatedly convert the formula into $\DNF$ (inside the scope of every modality), if we start with a formula of length $n$, the final translated formula has length $2^{O(n^2)}$. However, observe that the number of $\modules$ in the translated formula is linear in the size of the given formula $\phi$. Furthermore, the given formula is satisfiable in a model $M$ iff the translation is satisfiable in $M$ with appropriate modification of the $\rho$ (valuation function).

\subsection{Bounded agent property}
Now we prove that any formula $\theta \in \PTML^2$ in $\FSNF~\DNF$ is satisfiable iff $\theta$ is satisfiable in a model $M$ where the size of $D$ is bounded. Note that for any $\PTML$ formula $\theta$, if $M,w,\sigma \models \theta$ then $M^T,w,\sigma \models \theta$ where $M^T$ is the standard tree unravelling of $M$ with $w$ as root \cite{PR17}. Further, $M^T$ can be restricted to be of height at most $\md(\theta)$.  Hence, we restrict our attention to tree models of finite depth.

 First we define the notion of types for agents at every world. In classical $\FO^2$ the $2$-types are defined on atomic predicates. In $\PTML^2$ we need to define the types with respect to $\modules$. In any given tree model $M$ rooted at $r$, for any $w\in W$ and $c,d \in D_w$ the 2-$\type$ of $(c,d)$ at $w$ is simply the set of all $\modules$ that are true at $w$  where the two variables are assigned $c,d$ in either order. The 1-$\type$ of $c$ at $w$ includes the set of all $\modules$ that are true at $w$  when both $x,y$ are assigned $c$. Further, for every non-root node $w$, suppose ($w'\xrightarrow{a}w$) then the 1-$\type$ of any $c\in D_w$ should capture how $c$ behaves with respect to $a$ and the $1$-$\type(w,c)$ should also include the information of how $c$ acts with respect to $d$, for every $d\in D_w$. Thus the 1-$\type$ of $c$ at $w$ is given by a 3-tuple where the first component is the set of all $\modules$ that are true when both $x,y$ are assigned $c$, the second component captures how $c$ behaves with respect to the incoming edge of $w$ and the  third component is a set of subsets of formulas such that for each $d\in D_w$ there is a corresponding subset of formulas capturing the $2$-$\type$ of $c,d$.
To ensure that the type definition also carries the information of the height of the world $w$, if $w$ is at height $h$ then we restrict 1-$\type$ and 2-$\type$ at $w$ to $\modules$ of modal depth at most $\md(\phi) -h$.

\bigskip

For any formula $\phi$, let $\SF(\phi)$ be the set of all subformulas of $\phi$ closed under negation.  We always  assume\footnote{Let $p_0$ be some proposition occurring in $\phi$, then $\top$ is defined as $p_0 \lor \neg p_0$.} that $\top \in \SF(\phi)$. Let $\SF^h(\phi) \subseteq \SF(\phi)$ be the set of all subformulas of modal depth at most $\md(\phi)-h$. Thus we have $\SF(\phi) = \SF^0(\phi)  \supseteq \SF^1(\phi)\supseteq\ldots\supseteq \SF^{\md(\phi)}(\phi)$.

\begin{definition}[$\PTML~\type$]
\label{def: types} For any $\PTML^2$ formula $\phi$ and for any tree model $M$ rooted at $r$ with height at most $\md(\phi)$, for all $w\in W$ at height $h$:
\begin{itemize}
\item For all $c,d\in \delta(w)$,  define  2-$type(w,c,d) = (\Gamma_{xy};\Gamma_{yx})$ where\\
$\Gamma_{xy} = \{ \psi(x,y)  \in \SF^h(\phi) \mid M,w \models  \psi(c,d)\}$ and \\
$\Gamma_{yx} =  \{ \psi(x,y) \in \SF^h(\phi) \mid   M,w \models  \psi(d,c)\} $.

\item  If $w$ is a non root node, (say $w'\xrightarrow{a}w$) then for all $c\in \delta(w)$ define $1$-$\type(w,c) = (\Lambda_1;\Lambda_2;\Lambda_3)$ where 
$\Lambda_1 = 2$-$\type(w,c,c)$ and $\Lambda_2 = 2$-$\type(w,c,a)$ and 
$\Lambda_3 = \{2$-$\type(w,c,d) \mid d\in \delta(w)\}$.

\item For the root node $r$, for all $c\in \delta(r)$ define $1$-$\type(w,c) = (\Lambda_1;\{ \top\};\Lambda_3)$ where \\
$\Lambda_1 = 2$-$\type(w,c,c)$ and 
$\Lambda_3 = \{2$-$\type(w,c,d) \mid d\in \delta(w)\}$.
 \end{itemize}

\end{definition}

The second component of 1-$\type(r,c)$ is added to maintain uniformity. For all $w\in W$ define 1-$\type(w) = \{$1-$\type(w,c) \mid c\in D_w\}$ and 2-$\type(w) = \{ $2-$\type(w,c,d) \mid c,d \in D_w\}$. We use $\Lambda, \Pi$ to represent elements of 1-$\type(w)$ and $\Lambda_1,\Pi_2$ etc for the respective components.

\bigskip

If a formula $\theta$ is satisfiable in a tree model, the strategy is to inductively come up with bounded agent models for every subtree of the given tree (based on types), starting from leaves to the root.  While doing this, when we add new type based agents to a world at height $h$, to maintain monotonicity, we need to propagate the newly added agents throughout its descendants. For this, we define the notion of  extending any tree model by addition of some new set of agents.

Suppose in a tree model $M$, world $w$ has local agent set $D_w$ and we want to extend $D_w$ to $D_w \cup C$, then first we have $\Omega : C \mapsto D_w$ which assigns every new agent to some already existing agent. The intended meaning is that  the newly added agent $c\in C$ at $w$ mimics the  `type'  of $\Omega(c)$.
  If  $w$ is a leaf node, we can simply extend $\delta(w)$ to $D_w\cup C$. If $w$ is at some arbitrary height, along with \emph{adding the new agents to the live agent set} to $w$, we also need to create successors for every $c\in C$, one for each successor subtree of $\Omega(c)$ and inductively add $C$ to all the successor subtrees. 

\begin{definition}[Model extension]
\label{def: model extension}
Suppose $M$ is a tree model rooted at $r$ with finite agent set $D$ and for every $w\in W$ let  $M^w$ be the subtree rooted at $w$. Let $C$ be some finite set such that $C\cap D = \emptyset$ and for any $w\in W$ let   $\Omega: C \mapsto D_w$ be a function mapping $C$ to agent set live at $w$. Define the operation of  `adding $C$ to $M^w$ guided by $\Omega$'  by induction on the height of $w$ to obtain a new subtree rooted at $w$ (denoted by $M^w_{(C,\Omega)}$ and  the components denoted by $\delta', \rho'$ etc).
\begin{itemize}
\item If $w$ is a leaf, then $M^w_{(C,\Omega)}$ is a tree with a single node $w$ with new $\delta'(w) = \delta(w) \cup C$ and $\rho'(w) = \rho(w)$.
\item If  $w$ is at height $h$ then the new tree $M^w_{(C,\Omega)}$ is obtained from $M^w$ rooted at $w$ with new $\delta'(w) = \delta(w) \cup C$ and $\rho'(w) = \rho(w)$  and replacing all the subtrees $M^u$ rooted at every successor  $u$ of $w$ by $M^u_{(C,\Omega)}$. Furthermore, for every $c\in C$ and every $(w,\Omega(c),u) \in R$ create a new copy of $M^u_{(C,\Omega)}$ and rename its root as $u^c$ and add an edge $(w,c,u^c)$ to $R'$.
\end{itemize}
\end{definition}

 Since we do not have {\em equality} in the language, this transformation will still continue to satisfy the same formulas.

\begin{lemma}
\label{lemma: adding new agents is ok}
Let $M$ be any tree model of finite depth rooted at $r$ with finite agent set $D$ and let $w\in W$. Let  $M^w_{(C,\Omega)}$ (rooted at $w$) be an appropriate model extension of $M^w$ (rooted at $w$). For any interpretation $\sigma : \Var \mapsto (C\cup D_w)$ let $ \hat{\sigma}: \Var \mapsto D_w$ where $\hat{\sigma}(x) = \Omega(\sigma(x))$ if $\sigma(x) \in C$ and $\hat{\sigma}(x) = \sigma(x)$ if $\sigma(x) \in D_w$. Then for all $u\in W$ which is a descendant of $w$ in $M$ and for all $\sigma: \Var \mapsto (C\cup D_w)$ and for all $\PTML$ formula $\phi$, we have $M^w_{(C,\Omega)},u,\sigma \models \phi$  iff $M,u,\hat{\sigma} \models \phi$.

\end{lemma}

\begin{proof}
The proof is by reverse induction on the height of $w$. In the base case $w$ is a leaf. Note that $\rho(w)$ remains the same both the models. Hence all propositional formulas continue to equi-satisfy in both the models at $w$. Since $w$ is a leaf, there are no descendants in both the models and hence all modal formulas continue to equi-satisfy. Finally, since $\delta$ is non-empty in both the models at $w$,  for all formula $\alpha \in \PTML$ we have $M^w_{(\C,\Omega)},w,\sigma \models Q~ x~\alpha$  iff $M,w,\hat{\sigma} \models Q~ x~\alpha$ where for $Q\in \{\forall,\exists\}$.

\medskip
For the induction step, let $w$ be at height $h$. Now we induct on the structure of $\phi$. Again, if $\phi$ is a proposition, then the claim follows since $\rho(w)$ remains same. The cases of $\neg$ and $\land$ are standard.

For the case of $\Diamond_x ~\phi$, we need to consider two cases: when $\sigma(x) \in \C$ and $\sigma(x) \in D_w$.
\begin{itemize}
\item If $\sigma(x) \in \C$ then let $\Omega(c) = d$ and hence $\hat{\sigma}(x) = d$. Now, if
$M^w_{(\C,\Omega)},w,\sigma \models \Diamond_x \phi$ then there is some $(w,c,w') \in R^w_{(\C,\Omega)}$ such that $M^w_{(C,\Omega)},w',\sigma \models \phi$. By construction, $w'$ is of the form $u^c$ and the subtree rooted at $u^c$ is a copy of $M^u_{(\C,\Omega)}$ for some $(w,d,u) \in R$. Hence 
$M^u_{(\C,\Omega)},u,\sigma \models \phi$  and by induction hypothesis $M,u,\hat{\sigma} \models \phi$. Thus, $M,w,\hat{\sigma} \models \Diamond_x \phi$.

Suppose $M,w,\hat{\sigma} \models \Diamond_x \phi$, then there is some $(w,d,u)\in R$ such that $M,u,\hat{\sigma} \models \phi$.  By induction hypothesis, $M^u_{(\C,\Omega)},u,\sigma \models \phi$ . Now, since $\Omega(c) = d$, by construction there is $(w,d,u^c) \in R^w_{(C,\Omega)}$ such that the sub-tree rooted at $u^c$ is a copy of $M^u_{(\C,\Omega)}$. Hence $M^w_{(\C,\Omega)},u^c,\sigma \models \phi$. Thus   $M^w_{(\C,\Omega)},w,\sigma \models \Diamond_x \phi$.

\item If $\sigma(x) \in D_w$, let $\sigma(x) = d$. Now 
$M^w_{(C,\Omega)},w,\sigma \models \Diamond_x\phi$ iff there is some $(w,d,u) \in R^w_{(\C,\Omega)}$ such that $M^w_{(C,\Omega)},u,\sigma \models \phi$ iff (by construction) $(w,d,u) \in R$ and the sub-tree rooted at $u$ in $M^w_{(C,\Omega)}$ is a copy of $M^u_{(\C,\Omega)}$ iff $M^u_{(\C,\Omega)},u,\sigma \models \phi$ iff (by induction)  $M,u,\hat{\sigma} \models \phi$ iff $M,w,\hat{\sigma} \models \Diamond_x \phi$.
\end{itemize}
 
 For the case of $\exists x~\phi$, we have $M^w_{(C,\Omega)},w,\sigma \models \exists x~\phi$ iff there is some $c\in \C\cup D_w$ such that $M^w_{(C,\Omega)},w,\sigma_{[x\mapsto c]}\models \phi$ iff (by induction) $M,w,\hat{\sigma}_{[x\mapsto c]} \models \phi$ iff $M,w,\hat{\sigma} \models \exists x~\phi$.
\end{proof}

\bigbreak

For any formula in the normal form, we use the same notations as in  Def. \ref{def: FSNF normal form}. For a given formula $\theta \in \PTML^2$ in $\FSNF~\DNF$ form, let $\AllDelta_\theta^x = \{ \exists y~\delta^x \in \SF(\theta)\}$. Similarly we have $\AllDelta_\theta^y = \{ \exists x~\delta^y \in \SF(\theta)\}$ and $\AllPsi_\theta = \{ \forall x \exists y ~\psi \in \SF(\phi)\}$.

For any tree model $M$, let $\junk \not\in D$. For every $w\in W$ and for all $\exists y~\delta \in \AllDelta_\theta^x$ let the function $g^w_\delta : D_w \mapsto D_w\cup \{\junk\}$ be a mapping such that $M,w \models \delta(c,g^w_\delta(c))$ and $g^w_\delta(c) = \junk$ only if there is no $d\in D_w$ such that $M,w \models \delta(c,d)$.
Similarly for all $\exists x~\delta \in \AllDelta_\theta^y$ let $h^w_\delta : D_w \mapsto D_w\cup \{\junk\}$ such that $M,w \models \delta(h^w_\delta(c),h)$ and $h^w_\delta(c) = \junk$ only if there is no $d\in D_w$ such that $M,w \models \delta(d,c)$. Again for all $\forall x \exists y~\psi \in \AllPsi_\theta$ let $f^w_\psi : D_w \mapsto D_w\cup \{\junk\}$ such that $M,w \models \psi(c,f^w_\psi(c))$ and $f^w_\psi(c) = \junk$ only if there is no $d\in D_w$ such that $M,w \models \psi(c,d)$. 

The functions $g,h,f$ provide the  witnesses at a world for every agent (if it exists) for the existential formulas respectively.

\bigskip

\begin{theorem}
\label{thm: 2 var PTML has bounded agent property}
Let $\theta \in \PTML^2$ be in an $\FSNF~\DNF$ sentence. Then $\theta$ is satisfiable iff $\theta$ is satisfiable in a model with bounded number of agents.
\end{theorem}

\begin{proof} It suffices to prove $(\Rightarrow)$. Let $M$ be a tree model of height at most $\md(\theta)$ rooted at $r$ such that $M,r \models \theta$.

   Let  $E_\theta = \AllDelta_\theta^x \cup \AllDelta_\theta^y\cup \AllPsi_\theta$ and hence $|\AllExists_\theta|  \le |\theta|$ (say $q$). Let $\AllExists_\theta = \{\chi_1,\ldots \chi_q\}$ be some enumeration. For  every $w\in W$ and $a\in \delta(w)$ let $\Wit(a) = \{ b_1\ldots b_q\}$ be the witnesses for $a$ where $b_i = g^w_\delta(c)$ if $\chi_i$ is of the form $\exists y~ \delta \in \AllDelta^x_\theta$ (similarly $b_i = h^w_\delta(c)$ or $b_i = f^w_\psi(c)$ corresponding to $\chi_i$ of the from $\exists x~\delta^y$ and $\forall x\exists y~\psi$ respectively). If $b_i = \junk$ then set $b_i = b$ for some arbitrary but fixed $b\in \delta(w)$.

\medskip
For all $w\in W$ and $\Lambda \in~$1-$\type(w)$ fix some $a^w_\Lambda \in \delta(w)$ such that 1-$\type(w,a^w_\Lambda) = \Lambda$. Furthermore, if $c$ is the incoming edge of $w$ and 1-$\type(w,c) = \Lambda$ then let $a^w_\Lambda = c$. Let $A^w = \{ a^w_\Lambda \mid \Lambda \in$1-$\type(w)\}$.

\medskip

Now we define the bounded agent model. For every $w\in W$ let $M^w$ be the subtree model  rooted at  $w\in W$. For every such $M^w$, we define a corresponding {\em type based model} with respect to $\theta$ (denoted by $T^w_\theta$ with components denoted by $\delta^w_\theta, \rho^w_\theta$ etc) inductively as follows:

\begin{itemize}
\item  If $w$ is a leaf then $T^w_\theta$ is a tree with a single node $w$ with\\ $\delta^w_\theta(w) = $ 1-$\type(w) \times [1\ldots q] \times \{0,1,2\}$ and $\rho^w_\theta(w) = \rho(w)$.

\item If $w$ is at height $h$,  $T^w_\theta$ is a tree rooted at $w$ with $\delta^w_\theta(w) =$ 1-$\type(w) \times [1\ldots q] \times \{0,1,2\}$ and $\rho^w_\theta(w) = \rho(w)$.

Before defining the successors of $w$ in $T^w_\theta$ note that  for every $(w,a,u)\in R$ we have  $ T^u_\theta$ which is  the inductively constructed type based model rooted at $u$.  
Also,  inductively we have $\delta^u_\theta(u) =$1-$\type(u)\times [1\ldots q] \times \{0,1,2\}$.
\medskip

Now for every $a^w_\Lambda\in A^w$ let  $\{b_1\ldots b_q\}$ be the corresponding witnesses as described above. For every successor $(w,a^w_\Lambda,u) \in R$ and  for every $1\le e\le q$ and $f\in\{0,1,2\}$, create a new copy of $T^u_\theta$ (call it $N^{(\Lambda,e,f)}$) and name its root as $u^{(\Lambda,e,f)}$.  Now add $\delta^w_\theta(w)$ to $N^{(\Lambda,e,f)}$ at $u^{(\Lambda,e,f)}$  guided by $\Omega$ where $\Omega$ is defined as follows:
\begin{itemize}
\item For all $\Pi \in 1$-$\type(w)$ we have $a^w_\Pi \in A^w$. Define $\Omega((\Pi,e,f)) = ($1-$\type(u,a^w_\Pi),e,f)$.
\item for all $k \le q$ if 1-$\type(u,b_k) = \Pi$ then $\Omega((\Pi,k,f')) = ($1-$\type(u,b_k),e,f)$\\ where $f' = f+1 \mod 3$. 

\item  Let $f' = f-1 \mod 3$. For all $\Pi \in ~$1-$\type(w)$ let the witness set of $a^w_\Pi$ be $\{d_1\ldots d_q\}$.

For all $l\le q$ if $1$-$\type(w,d_{l}) = \Lambda$
then by $\Lambda_3$ component, there is some $a\in \delta(w)$ such that $2$-$\type(w,d_{l},a^w_\Pi) = 2$-$\type(w,a^w_\Lambda,a)$. \ \
Define $\Omega((\Pi,l,f')) =($1-$\type(u,a),e,f)$.

\item For all $(\Pi,e',f')\in \delta^w_\theta(w)$ if $\Omega(\Pi,e',f')$  is not  yet defined, then set $\Omega(\Pi,e',f') = ($1-$\type(u,a^w_\Pi),e,f)$.
\end{itemize}

Add an edge  $(w, (\Lambda,e,f),u^{(\Lambda,e,f)})$ to $R^w_\theta$. 

\end{itemize}

Note that $\Omega$ is well defined since the first three steps are defined  for the indices $f, (f$+$1 \mod 3)$ and  $(f$-$1 \mod 3)$ respectively, which are always distinct. Also note that $T^r_\theta$ is a model that satisfies bounded agent property. Thus, it is sufficient to prove that $T^r_\theta,r \models \theta$.

 \bigskip

{\bf Claim.} For every $w\in W$ at height $h$ and for all $\lambda \in \SF^h(\theta)$  the following holds:
\begin{enumerate}
\item Suppose $\lambda$ is a sentence and $M,w \models \lambda$ then $T^w_\theta,w \models \lambda$.

\item If $\FV(\lambda) \subseteq \{x,y\}$ and for all $\Lambda,\Pi \in~$1-$\type(w)$ if     $M,w,[x\mapsto a^w_\Lambda,y \mapsto a^w_\Pi] \models  \lambda$  then  for all $1\le e \le q$ and $f \in \{0,1,2\}$ we have  $T^w_\theta,w,[x\mapsto (\Lambda,e,f), y\mapsto (\Pi,e,f)] \models \lambda$.
\end{enumerate}

 \medskip
Note that the theorem follows from claim (1), since $\theta$ is sentence and $M,r\models\theta$. 
\medskip

The proof of the claim is by reverse induction on $h$.  In the base case $h = \md(\theta)$ which implies $\lambda$ is modal free and hence is a $\DNF$ over $\literals$. Thus, both the claims follow since $\rho(w) = \rho^w_\theta(w)$.

For the induction step, let $w$ be at height $h$. Now we induct on the structure of $\lambda$. Again if $\lambda$ is a $\literal$ then both the the claims follow since $\rho(w) = \rho^w_\theta(w)$. The case of $\land$ and $\lor$ are standard.

For the case $\Box_x \lambda$, we only need to prove claim(2).  
Now suppose $M,w,[x\mapsto a^w_\Lambda,y\mapsto a^w_\Pi] \models \Box_x \lambda$. Pick arbitrary $e$ and $f$. We need to prove that $T^w_\theta,w,[x\mapsto (\Lambda,e,f), y\mapsto (\Pi,e,f)] \models \Box_x \lambda$. Pick any $(w,(\Lambda,e,f),u^{(\Lambda,e,f)}) \in R^w_\theta$, then by construction we have $(w,a^w_\Lambda,u)\in R$ and since $M,w,[x\mapsto a^w_\Lambda,y\mapsto a^w_\Pi] \models \Box_x \lambda$, we have $M,u,[x\mapsto a^w_\Lambda,y\mapsto a^w_\Pi] \models \lambda$. Let $ a^u_{\Pi'} \in A^u$ such that  1-$\type(u,a^u_{\Pi'}) =$1-$\type(u,a^w_\Pi)$ and since $a^w_\Lambda$ is the incoming edge of $u$,  by $\Pi_2$ component, we have 2-$\type(u,a^w_\Pi,a^w_\Lambda) = $2-$\type(u,a^u_{\Pi'},a^w_\Lambda)$ and also $a^w_\Lambda \in A^u$ . Hence $M,u, [x\mapsto a^w_\Lambda, y \mapsto a^u_{\Pi'}] \models \lambda$ and by induction  hypothesis $T^u_\theta,u,[x\mapsto ($1-$\type(u,a^w_\Lambda),e,f), y\mapsto ($1-$\type(u,a^u_{\Pi'}),e,f)] \models \lambda$. Now by construction, at $u^{(\Lambda,e,f)}$ we have $\Omega{(\Lambda,e,f)} = ($1-$\type(w,a^w_\Lambda),e,f)$ and $\Omega(\Pi,e,f) = ($1-$\type(u,a^u_{\Pi'}),e,f)$. Thus, by Lemma \ref{lemma: adding new agents is ok}, $T^w_\theta,u^{(\Lambda,e,f)},[x\mapsto (\Lambda,e,f), y\mapsto (\Pi,e,f)] \models \lambda$. Hence, we have $T^w_\theta,w,[x\mapsto (\Lambda,e,f), y\mapsto (\Pi,e,f)] \models \Box_x \lambda$.\ \
The case for $\Box_y \lambda$ is analogous.

\medskip
For the case $\Diamond_y \lambda$, again only claim(2) applies. Suppose $M,w,[x\mapsto a^w_\Lambda,y\mapsto a^w_\Pi] \models \Diamond_y \lambda$.  Now pick $e$ and $f$ appropriately. We need to prove that $T^w_\theta,w,[x\mapsto (\Gamma,e,f), y\mapsto (\Pi,e,f)] \models \Diamond_y \lambda$.
By supposition, there is some $w\xrightarrow{a^w_\Pi}u$ such that $M,u,[x\mapsto a^w_\Lambda,y\mapsto a^w_\Pi] \models \lambda$.
Using the argument similar  to the previous case, we can prove that $T^w_\theta,u^{(\Lambda,e,f)},[x\mapsto (\Lambda,e,f), y\mapsto (\Pi,e,f)] \models \lambda$ and hence $T^w_\theta,w,[x\mapsto (\Gamma,e,f), y\mapsto (\Pi,e,f)] \models \Diamond_y \lambda$. \ \
The case of $\Diamond_x \lambda$ is symmetric.

\bigskip

For the case $\exists y~\lambda$ (where $x$ is free at the outer most level), for claim (2) first note that since $\theta$ is in the normal form, $\lambda$ is $\Qsafe$. Also note that $\exists y~\lambda = \chi_i$ for some $\chi_i\in E_\theta$. Now, suppose $M,w,[x\mapsto a^w_\Lambda] \models \exists y~\lambda$ then we need to prove that $T^w_\theta,w,[x\mapsto (\Lambda,e,f)] \models \exists y~\lambda$. Let the $i^{th}$ witness of $a^w_\Lambda$ be $b_i$ and hence $M,w,[x\mapsto a^w_\Lambda, y\mapsto b_i] \models \lambda$. Let 1-$\type(w,b_i) = \Pi'$, we claim that $T^w_\theta,w,[x\mapsto(\Lambda,e,f), y\mapsto (\Pi',i,f')] \models \lambda$ where $f' = f+1 \mod 3$. Suppose not, then $\land$ and $\lor$ can be broken down and we get some $\module$ such that $M,w,[x\mapsto a^w_\Lambda, y\mapsto b_i] \models \Delta_z  \lambda'$ and $T^w_\theta,w,[x\mapsto(\Lambda,e,f), y\mapsto (\Pi',i,f')] \not\models \Delta_z \lambda'$ where $\Delta\in \{\Box,\Diamond\}$ and $z\in \{x,y\}$. Assume $\Delta = \Box$ and $z=x$ (other cases are analogous). This implies $T^w_\theta,w,[x\mapsto(\Lambda,e,f), y\mapsto (\Pi',i,f')] \models \Diamond_x \neg \lambda'$ and hence there is some $w\xrightarrow{(\Lambda,e,f)}u^{(\Lambda,e,f)}$ such that $T^w_\theta, u^{(\Lambda,e,f)}, [x\mapsto(\Lambda,e,f), y\mapsto (\Pi',i,f')] \models \neg \lambda'$(*).  By construction, there is a corresponding $w\xrightarrow{a^w_\Lambda}u$ in $M$. Now since $M,w,[x\mapsto a^w_\Lambda, y\mapsto b_i] \models  \Box_x  \lambda'$, we have $M,u,[x\mapsto a^w_\Lambda, y\mapsto b_i] \models  \lambda'$. Let $b_i' \in A^u$ such that 1-$\type(u,b_i) =$1-$\type(u,b_i')$. Since  $a^w_\Lambda$ is the incoming edge to $u$ by $\Pi'_2$ component, we have 2-$\type(u,b_i,a^w_\Lambda) =$2-$\type(u,b_i',a^w_\Lambda)$ and $a^w_\Lambda \in A^u$. Thus,  $M,u,[x\mapsto a^w_\Lambda, y\mapsto b'_i] \models  \lambda'$ and by induction hypothesis,  $T^u_\theta,u, [x\mapsto(\Lambda,e,f), y\mapsto (1$-$\type(u,b_i'),e,f)] \models  \lambda'$. Again by construction, at $u$ we have $\Omega((\Lambda,e,f)) = (\Lambda,e,f)$ and $\Omega((\Pi',i,f')) = (1$-$\type(u,b_i'),e,f)$ and hence  by Lemma \ref{lemma: adding new agents is ok}, $T^w_\theta, u^{(\Lambda,e,f)}, [x\mapsto(\Lambda,e,f), y\mapsto (\Pi',i,f')] \models  \lambda'$ which is a contradiction to (*).  \ \ \ 
The case of $\exists y~\lambda$ is analogous.

\medskip

For the case of $\forall x~\lambda$ (where $y$ is free at the outer most level), suppose $M,w,[y\mapsto a^w_\Pi] \models \forall x~\lambda$. We need to prove that $T^w_\theta,w,[y\mapsto (\Pi,e,f)]\models \forall x~\lambda$. Pick any $(\Lambda',e',f')\in \delta^w_\theta(w)$, now we claim $T^w_\theta, w, [x\mapsto (\Lambda',e',f'), y\mapsto (\Pi,e,f)] \models \lambda$ (otherwise, like in the previous case, since $\lambda$ is $\Qsafe$, we can reach a $\module$ where they differ and obtain a contradiction). \ \ \ 
The case $\forall y~\lambda$ is analogous.

\bigskip

Finally we come to sentences which are relevant for claim (1). Note that in the normal form, at the outermost level, a sentence will have only $\literals$ or formulas of the form $\forall x \exists y~\psi_l$ or $\forall x \forall y~\phi$. 

For the case $M,w \models \forall x \exists y~ \psi_l$, let $\forall x\exists y~\psi_l$ be $i^{th}$ formula in $E_\theta$. We need to prove $T^w_\theta,w \models \forall x \exists y~\psi_l$. Pick any $(\Lambda,e,f)\in \delta^w_\theta(w)$ and we have $a^w_\Lambda \in A^w$. Let the $i^{th}$ witness for $a^w_\Lambda$ be $b_i$. Thus we have $M,w,[x\mapsto a_\Gamma, y\mapsto b_i] \models \psi_l$.  Let $1$-$\type(w,b_i) = \Pi'$. Again we claim that $T^w_\theta,w,[x\mapsto (\Gamma,e,f), y\mapsto[\Pi',e,f')] \models \psi_l$ where $f' = f+1 \mod 3$. Suppose not, again $\land$ and $\lor$ can be broken down and we get some $\module$ such that $M,w,[x\mapsto a^w_\Lambda, y\mapsto b_i] \models \Delta_z  \lambda'$ and $T^w_\theta,w,[x\mapsto(\Lambda,e,f), y\mapsto (\Pi',i,f')] \not\models \Delta_z \lambda'$ where $\Delta\in \{\Box,\Diamond\}$ and $z\in \{x,y\}$. Assume $\Delta = \Diamond$ and $z=y$ (other cases are analogous). This implies $T^w_\theta,w,[x\mapsto(\Lambda,e,f), y\mapsto (\Pi',i,f')] \models \Box_y \neg \lambda'$ (*). 
Now let $a^w_{\Pi'} \in A^w$ such that 1-$\type(w,a^w_{\Pi'}) = 1$-$\type(w,b_i) = \Pi'$. Thus by $\Pi'_3$ component, there is some $d \in \delta^w_{\theta}$ such that 2-$\type(w,a^w_{\Pi'},d) = 2$-$\type(w,b_i,a^w_\Lambda)$ and hence $M,w,[x\mapsto d, y\mapsto a^w_{\Pi'}] \models \Diamond_y \lambda'$. Hence there is some $w\xrightarrow{a^w_{\Pi'}}u$ such that $M,u, [x\mapsto d,y\mapsto a^w_{\Pi'}] \models \lambda'$. Now let 1-$\type(u,d) = 1$-$\type(u,d')$ such that $d' \in A^u$ and since $a^w_{\Pi'}$ is the incoming edge, we have $M,u, [x\mapsto d',y\mapsto a^w_{\Pi'}] \models \lambda'$ and by induction hypothesis, $T^u_\theta,u, [x\mapsto (1$-$\type(u,d'),i,f'), y\mapsto (1$-$\type(u,a^w_{\Pi'}),i,f')] \models \lambda'$ and while constructing $u^{(\Pi',i,f')}$ (case 3 applies for $a^w_\Lambda$ since its $i^{th}$ witness has same 1-$\type$ as $a^w_{\Pi'}$) we have $\Omega((\Lambda,e,f'-1)) = (1$-$\type(u,d'),i,f')$. Thus by Lemma \ref{lemma: adding new agents is ok} (since $f'-1 = f$), $T^w_\theta,u^{(\Pi',i,f')},[x\mapsto(\Lambda,e,f), y\mapsto (\Pi',i,f')] \models \lambda'$ which contradicts (*).
\medskip

Finally, for the case $\forall x \forall y~\phi$ suppose $M,w \models \forall x \forall y~\phi$, then for any $(\Gamma,e,f),~ (\Delta,e',f') \in \delta^w_\theta(w)$ we claim that $T^w_\theta,w,[x\mapsto (\Gamma,e,f), y\mapsto (\Delta,e',f')] \models \phi$ (else again, go to the smallest module and prove contradiction).
\end{proof}

Note that in the type based model, at any world $w$ we have $|\delta^w_\theta|  = 2^{2^{O(|\SF(\theta)|)}}$. Now if we start with a $\PTML^2$ formula $\phi$, then though its corresponding equi-satisfiable formula $\theta$ is exponentially larger, the number of distinct subformulas in $\theta$ is still linear in the size of $\phi$.

\begin{corollary}
$\TML^2$ satisfiability is in 2-$\EXPSPACE$.
\end{corollary}
\begin{proof}
Any $\TML^2$ formula $\alpha$ is satisfiable iff (by Lemma.\ref{lem: TML and PTML satisfiability are same}) its corresponding $\PTML^2$ translation $\phi$ is satisfiable  iff (by Theorem \ref{thm: 2 var PTML has bounded agent property}) the corresponding normal form $\theta$  of $\phi$ is satisfiable over agent set $D$ of size $2^{2^{O(|\phi|)}}$ iff  (by Lemma. \ref{lem: inc and const domain sat are same}) $\hat{\theta} \in \PTML^2$ is satisfiable in a constant domain model over $D$.

Thus we can expand the quantifiers of $\hat{\theta}$ by corresponding $\bigwedge$ and $\bigvee$ for $\forall$ and $\exists$ respectively and we get a propositional multi-modal formula. This satisfiability is in $\PSPACE$. But in terms of the size of the formulas, $|\hat{\theta}| = 2^{2^{|\alpha|^2}}$. Thus we have a 2-$\EXPSPACE$ algorithm.
\end{proof}

\subsection{Example}

We illustrate the construction of {\em type based models} with an example.
Consider the $\PTML^2$ sentence $\theta := \forall x~\Box_x \Box_x \bot \land \forall x \exists y~(\Box_x (\Diamond_y (\neg p) \land \exists y~ \Diamond_y  p))$ which is in $\FSNF~\DNF$.
Let $M$ be the model described in Fig.~\ref{fig: example original model}  where 
\begin{itemize}
\item $W = \{r\} \cup \{ u^i,v^i,w^i \mid i \in \Nat\}$ 
\item $D = \Nat$
\item $\delta(r) = \{ 2i \mid i\in \Nat\}$ (all even numbers) and \\ $\delta(w^i) = \delta(u^i) = \delta(v^i) = \Nat$
\item $R = \{(r,2i,w^i), (w^i,2i+1,u^i), (w^i,2i+2,v^i) \mid i\in \Nat\}$
\item $\rho(r) = \rho(w^i) = \rho(v^i) = \emptyset$ and $\rho(u^i) = p$ for all $i\in \Nat$.
\end{itemize}

\begin{figure}[h]
\includegraphics[scale=0.07]{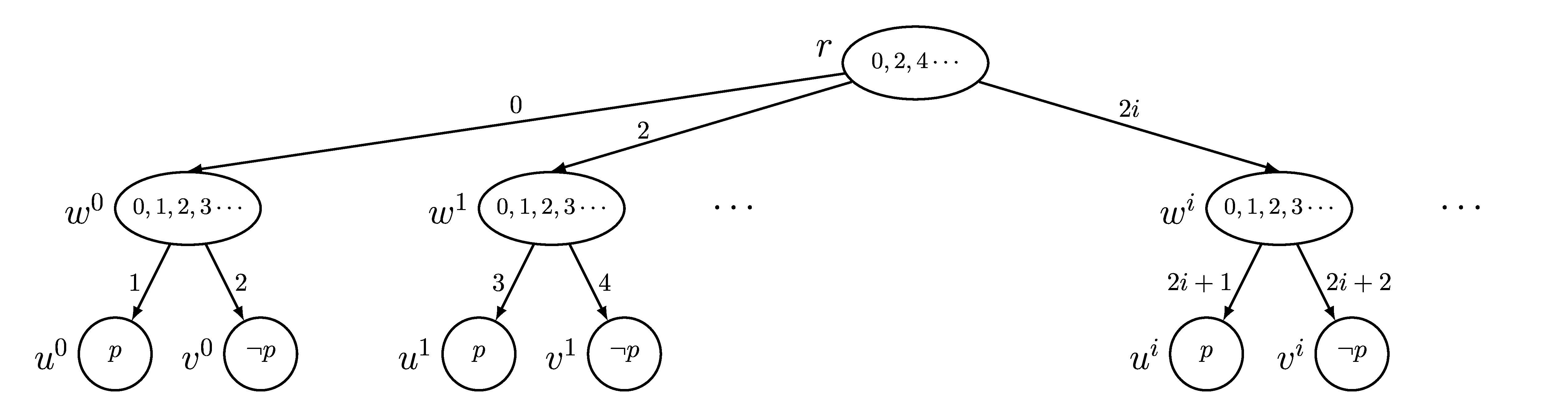}

	\caption{Given model such that $M,r \models \theta$.}
	\label{fig: example original model}
\end{figure}

\begin{figure}[h]
\includegraphics[scale=0.07]{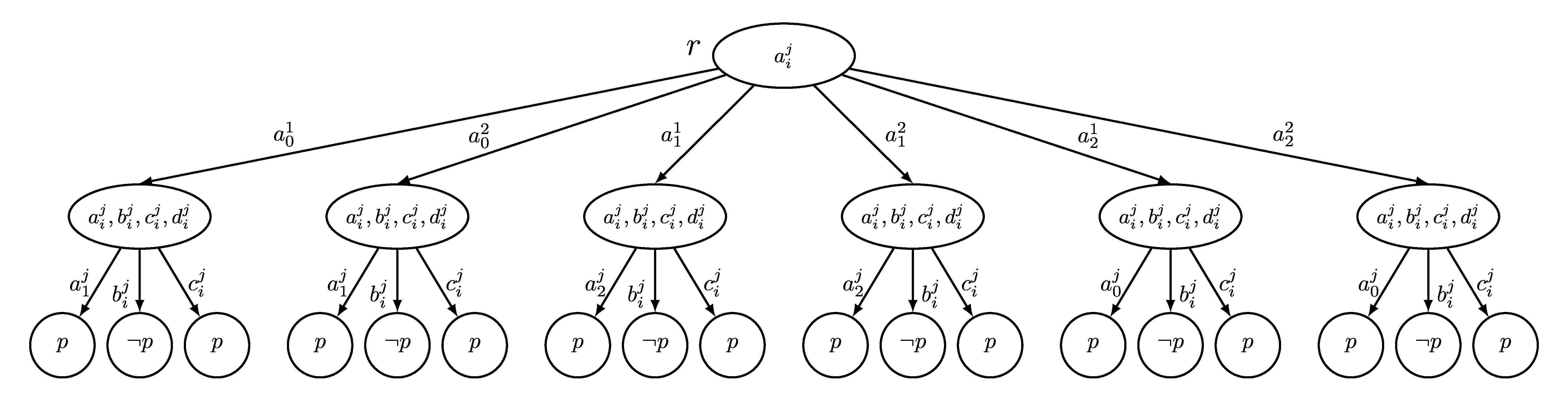}

	\caption{Corresponding bounded agent model with $M',r \models \theta$.  $a^j_i,b^j_i,c^j_i$ corresponds to agents with $1\le j\le 2$ and $i\in \{0,1,2\}$. The  edge $a^j_i,b^j_i, c^j_i$   indicate one successor for every $1\le j \le 2$ and $i\in \{0,1,2\}$.}
\label{fig: example bounded model}
\end{figure}

 Clearly,  $M,r \models \theta$. Let $f^r:D_r \mapsto D_r$ be defined by $f^r(2i) = 2i+2$ and at all $w^i$, $g^i(j) = 2i+1$ for all $i\in \Nat$ be the two (relevant) witness functions. The one and two types at every world are described as follows:

At leaf nodes $u^i$ and $v^i$ there is only one distinct one type and two types.
At $w^i$, note that $r\xrightarrow{2i}w_i$ is the incoming edge and only $2i+1$ and $2i+2$ have outgoing edges. Thus, there are $3$ distinct 1-$\type$ members at $w^i$, each for $(2i+1), (2i+2)$ and [the rest]. Let $b,c,d$ be the respective types. Finally at the root again we have only a single distinct type (call it $a$).  

\medskip
Since there are $2$ existential formulas, the root of the type based model has $(1\times 2\times 3) = 6$ agents let it be $\{a^e_f\mid 1\le e \le 2,~ 0\le f \le 2\}$ and $0$ be the representative.
At $w^0$ we have $(3\times 2 \times 3) = 18$ agents. Let the representatives be $1,2,0$ for $b,c,d$ respectively. Note that we cannot pick any other representative for [the rest] other than $0$ since $0$ is the incoming edge to $w^0$. Let the bounded agent set be $\{ b^e_f,~c^e_f,~d^e_f\mid 1\le e \le 2,~ 0\le f\le 2\}$. The corresponding bounded model $M'$ is described in Figure \ref{fig: example bounded model}. It can be verified that $M',r \models \theta$.

\section{Discussion}
We have proved that the two variable fragment of $\PTML^2$ (and hence $\TML^2$) is 
decidable. The upper bound shown is in 2-$\EXPSPACE$.  A $\NEXP$ lower bound 
follows since $\FO^2$ satisfiability can be reduced to $\PTML^2$ satisfiability. 
We believe that by careful management of the normal form, space can be reused
and the upper bound can in fact be brought down by one exponent. That would still leave a significant
gap between lower and upper bounds to be addressed in future work.

We can also prove that addition of constants makes $\PTML^2$ undecidable. In fact, with
the addition of a single constant $\mathbf{c}$ we can use $\Box_\mathbf{c}$ to simulate the `free' 
$\Box$ of $\FOML^2$, thus yielding undecidability.  When it comes to equality, the 
situation is more tricky: note that we can no longer use model extension 
(Def.\ref{def: model extension} and Lemma \ref{lemma: adding new agents is ok}) since equality might restrict the number of 
agents at every world.

The most important issue is expressiveness. What kind of accessibility relations
or model classes can be characterized by 2-variable $\TML$? This is unclear,
but there are sufficiently intriguing examples and applications making the
issue an interesting challenge.

\bibliography{ref}

\begin{thebibliography}{10}

\bibitem{guardedFO}
Hajnal Andr{\'e}ka, Istv{\'a}n N{\'e}meti, and Johan van Benthem.
\newblock Modal languages and bounded fragments of predicate logic.
\newblock {\em Journal of philosophical logic}, 27(3):217--274, 1998.

\bibitem{bluebook}
Patrick Blackburn, Maarten de~Rijke, and Yde Venema.
\newblock {\em Modal Logic (Cambridge Tracts in Theoretical Computer Science)}.
\newblock Cambridge University Press, 2001.
\newblock URL:
  \url{http://www.amazon.com/Cambridge-Tracts-Theoretical-Computer-Science/dp/0521802008%3FSubscriptionId%3D0JYN1NVW651KCA56C102%26tag%3Dtechkie-20%26linkCode%3Dxm2%26camp%3D2025%26creative%3D165953%26creativeASIN%3D0521802008}.

\bibitem{corsi2002unified}
Giovanna Corsi.
\newblock A unified completeness theorem for quantified modal logics.
\newblock {\em The Journal of Symbolic Logic}, 67(4):1483--1510, 2002.

\bibitem{ModalNormalForm}
Kit Fine et~al.
\newblock Normal forms in modal logic.
\newblock {\em Notre Dame journal of formal logic}, 16(2):229--237, 1975.

\bibitem{FOMLbook}
Melvin Fitting and Richard~L. Mendelsohn.
\newblock {\em First-Order Modal Logic (Synthese Library)}.
\newblock Springer, 1999.
\newblock URL:
  \url{http://www.amazon.com/First-Order-Modal-Logic-Synthese-Library/dp/0792353358%3FSubscriptionId%3D0JYN1NVW651KCA56C102%26tag%3Dtechkie-20%26linkCode%3Dxm2%26camp%3D2025%26creative%3D165953%26creativeASIN%3D0792353358}.

\bibitem{TML2001}
Melvin Fitting, Lars Thalmann, and Andrei Voronkov.
\newblock Term-modal logics.
\newblock {\em Studia Logica}, 69(1):133--169, 2001.
\newblock URL: \url{http://dx.doi.org/10.1023/A:1013842612702}, \href
  {http://dx.doi.org/10.1023/A:1013842612702}
  {\path{doi:10.1023/A:1013842612702}}.

\bibitem{FO2}
Erich Gr{\"a}del, Phokion~G Kolaitis, and Moshe~Y Vardi.
\newblock On the decision problem for two-variable first-order logic.
\newblock {\em Bulletin of symbolic logic}, 3(1):53--69, 1997.

\bibitem{FO2extension}
Erich Gr{\"a}del and Martin Otto.
\newblock On logics with two variables.
\newblock {\em Theoretical computer science}, 224(1-2):73--113, 1999.

\bibitem{hughes96}
MJ~Hughes and GE~Cresswell.
\newblock {\em A New Introduction to Modal Logic. Routledge. 1996.}
\newblock Routledge, 1996.
\newblock URL:
  \url{http://www.amazon.com/Introduction-Modal-Logic-Routledge-1996/dp/B00DHNF0MU%3FSubscriptionId%3D0JYN1NVW651KCA56C102%26tag%3Dtechkie-20%26linkCode%3Dxm2%26camp%3D2025%26creative%3D165953%26creativeASIN%3DB00DHNF0MU}.

\bibitem{FOML2var}
Roman Kontchakov, Agi Kurucz, and Michael Zakharyaschev.
\newblock Undecidability of first-order intuitionistic and modal logics with
  two variables.
\newblock {\em Bulletin of Symbolic Logic}, 11(3):428--438, 2005.

\bibitem{kooi2007}
Barteld Kooi.
\newblock Dynamic term-modal logic.
\newblock In {\em A Meeting of the Minds}, pages 173--186, 2007.

\bibitem{Kripke}
Saul~A. Kripke.
\newblock The undecidability of monadic modal quantification theory.
\newblock {\em Mathematical Logic Quarterly}, 8(2):113--116, 1962.
\newblock URL: \url{http://dx.doi.org/10.1002/malq.19620080204}, \href
  {http://dx.doi.org/10.1002/malq.19620080204}
  {\path{doi:10.1002/malq.19620080204}}.

\bibitem{mortimer2var}
Michael Mortimer.
\newblock On languages with two variables.
\newblock {\em Mathematical Logic Quarterly}, 21(1):135--140, 1975.

\bibitem{orlandelli}
Eugenio Orlandelli and Giovanna Corsi.
\newblock Decidable term-modal logics.
\newblock In {\em 15th European Conference on Multi-Agent Systems}, 2017.

\bibitem{PR17}
Anantha Padmanabha and R~Ramanujam.
\newblock The monodic fragment of propositional term modal logic.
\newblock {\em Studia Logica}, pages 1--25, 2018.

\bibitem{TML2varArxiv}
Anantha Padmanabha and R~Ramanujam.
\newblock Two variable fragment of term modal logic.
\newblock {\em arXiv preprint arXiv:1904.10260}, 2019.

\bibitem{PRW18}
Anantha Padmanabha, R~Ramanujam, and Yanjing Wang.
\newblock {Bundled Fragments of First-Order Modal Logic: (Un)Decidability}.
\newblock In Sumit Ganguly and Paritosh Pandya, editors, {\em 38th IARCS Annual
  Conference on Foundations of Software Technology and Theoretical Computer
  Science (FSTTCS 2018)}, volume 122 of {\em Leibniz International Proceedings
  in Informatics (LIPIcs)}, pages 43:1--43:20, Dagstuhl, Germany, 2018. Schloss
  Dagstuhl--Leibniz-Zentrum fuer Informatik.
\newblock URL: \url{http://drops.dagstuhl.de/opus/volltexte/2018/9942}, \href
  {http://dx.doi.org/10.4230/LIPIcs.FSTTCS.2018.43}
  {\path{doi:10.4230/LIPIcs.FSTTCS.2018.43}}.

\bibitem{rybakov2017}
Mikhail Rybakov and Dmitry Shkatov.
\newblock Undecidability of first-order modal and intuitionistic logics with
  two variables and one monadic predicate letter.
\newblock {\em Studia Logica}, pages 1--23, 2017.

\bibitem{shtakser2018}
Gennady Shtakser.
\newblock Propositional epistemic logics with quantification over agents of
  knowledge.
\newblock {\em Studia Logica}, 106(2):311--344, 2018.

\bibitem{Shtakser2018b}
Gennady Shtakser.
\newblock Propositional epistemic logics with quantification over agents of
  knowledge (an alternative approach).
\newblock {\em Studia Logica}, Aug 2018.
\newblock URL: \url{https://doi.org/10.1007/s11225-018-9824-6}, \href
  {http://dx.doi.org/10.1007/s11225-018-9824-6}
  {\path{doi:10.1007/s11225-018-9824-6}}.

\bibitem{Wang17}
Yanjing Wang.
\newblock A new modal framework for epistemic logic.
\newblock In {\em Proceedings Sixteenth Conference on Theoretical Aspects of
  Rationality and Knowledge, {TARK} 2017, Liverpool, UK, 24-26 July 2017.},
  pages 515--534, 2017.
\newblock URL: \url{https://doi.org/10.4204/EPTCS.251.38}, \href
  {http://dx.doi.org/10.4204/EPTCS.251.38} {\path{doi:10.4204/EPTCS.251.38}}.

\bibitem{WangTML}
Yanjing Wang and Jeremy Seligman.
\newblock When names are not commonly known: Epistemic logic with assignments.
\newblock In {\em Advances in Modal Logic Vol. 12 (2018): 611-628, College
  Publications}.

\bibitem{FOMLdecidable}
Frank Wolter and Michael Zakharyaschev.
\newblock Decidable fragments of first-order modal logics.
\newblock {\em The Journal of Symbolic Logic}, 66(3):1415--1438, 2001.

\end{thebibliography}

\end{document}